\documentclass[journal,10pt,onecolumn]{IEEEtran}

\usepackage{amsmath,amssymb,amsthm}
\usepackage{mathtools}
\usepackage{bm}
\usepackage{hyperref}
\usepackage{cite}
\usepackage{graphicx}
\usepackage{xcolor}
\usepackage{bbm}
\usepackage{booktabs}
\usepackage{array}

\newtheorem{theorem}{Theorem}
\newtheorem{lemma}{Lemma}
\newtheorem{corollary}{Corollary}

\newtheorem{definition}{Definition}
\newtheorem{remark}{Remark}
\newtheorem{example}{Example}

\newcommand{\Fb}{\mathbb{F}}
\newcommand{\Eb}{\mathbb{E}}
\newcommand{\Pb}{\mathbb{P}}
\newcommand{\wH}{w_{\mathrm{H}}}
\newcommand{\Nb}{\mathcal{N}}
\newcommand{\Cb}{\mathcal{C}}
\newcommand{\Hb}{h_b}
\newcommand{\Seta}{\mathcal{S}}
\newcommand{\eps}{\varepsilon}
\newcommand{\Rb}{\mathcal{R}}

\newcommand{\Hm}{\mathbf{H}}
\newcommand{\Harm}{\mathcal{H}}

\begin{document}

\title{Guesswork Under Linear Constraints: \\
Exact Exponent for Coset Decoding}

\author{Hassan Tavakoli,\\
School of EECS, \\
Oregon State University, \\
tavakolh@oregonstate.edu}

\maketitle

%------------------------------------------------------------------
%------------------------------------------------------------------
\begin{abstract}
We establish the exact exponential growth rate of the $\rho$-th moment
of the constrained guesswork $G_{\mathrm{coset}}$---the rank of the
true noise vector within its syndrome coset of a random binary linear
code under i.i.d.\ Bernoulli$(p)$ noise:
\(
    \lim_{n\to\infty}
    \frac{1}{n}\log_2\Eb\!\left[G_{\mathrm{coset}}^{\rho}\right]
    \;=\;
    \rho\,h_{\frac{1}{1+\rho}}(p)\;+\;\rho(R-1),
    \, \rho>0,
\)
where $h_\alpha(p)$ is the binary R\'{e}nyi entropy and $R=k/n$
is the code rate.
The exponent shifts down by exactly $\rho(1-R)$ relative to the
unconstrained Ar{\i}kan--Merhav exponent, with each of the $n(1-R)$
parity checks contributing equally.
Finite-length simulations confirm convergence from below.
We further establish: (i)~a transfer theorem expressing the
partition-function exponent in terms of an arbitrary weight-enumerator
growth rate $g(\delta)$; (ii)~the exact exponent for $L_n$-list
(``$k$-th'') constrained guesswork; and (iii)~a sharp second-order
refinement of order $\rho\log_2 n$.
Beyond the binary i.i.d.\ setting, we prove
a universality theorem: for any code ensemble $\mathcal{E}$ whose
weight enumerator concentrates at rate $g_{\mathcal{E}}(\delta)$,
the guesswork exponent equals
$(1+\rho)\psi_{1/(1+\rho)}(g_{\mathcal{E}})-\rho\,\psi_1(g_{\mathcal{E}})$,
where $\psi_\alpha(g)=\sup_\delta[g(\delta)+\alpha\ell(\delta)]$.
As concrete applications, we instantiate this theorem for the
$q$-ary extension,
$\Lambda_q(\rho)=\rho\,h^{(q)}_{1/(1+\rho)}(P)+\rho(R-1)\log_2 q$,
and for Gallager's regular LDPC ensemble, obtaining a closed-form
guesswork exponent via an exact finite-length identity for the
ensemble-average weight enumerator.
\end{abstract}

\begin{IEEEkeywords}
Guesswork, GRAND decoding, random linear codes,
R\'{e}nyi entropy, guesswork exponent, coset enumeration,
transfer theorem, list guesswork, second-order exponent,
universality theorem, $q$-ary guesswork, LDPC ensemble.
\end{IEEEkeywords}

%==================================================================
\section{Introduction}
\label{sec:intro}
%==================================================================

The guesswork of a random variable $X$ introduced by
Massey~\cite{massey1994guessing} and quantified by
Ar{\i}kan~\cite{arikan1996inequality} and
Ar{\i}kan--Merhav~\cite{arikan1998guessing}
counts how many guesses an optimal strategy requires to identify
a realization of~$X$.
For an i.i.d.\ source $X^n\sim P_X^{\otimes n}$:
\(
    \lim_{n\to\infty}\frac{1}{n}\log\Eb[G(X^n)^\rho]
    \;=\;
    \rho\,h_{\!\frac{1}{1+\rho}}(X),
\)
an exact equality proved in~\cite{arikan1996inequality}.
The R\'{e}nyi entropy $h_\alpha(X)$ at order $\alpha=1/(1+\rho)<1$
thus governs the exponential growth rate of the $\rho$-th guesswork
moment.
Connections to large deviations and channel coding have been explored
in~\cite{christiansen2013guessing,merhav2020guessing}.
Guessing Random Additive Noise Decoding
(GRAND)~\cite{duffy2021GRAND} decodes by querying
noise patterns $e'$ in decreasing-probability order, testing
$\Hm (y\oplus e')^T=\bm{0}$ at each step.
Its query complexity $G_{\mathrm{GRAND}}$ ranges over all of
$\{0,1\}^n$.

The paper is organized as follows.
Section~\ref{sec:model} establishes notation.
Section~\ref{sec:sandwich} proves the sandwich inequality.
Section~\ref{sec:spectrum} proves the uniform spectrum law.
Section~\ref{sec:pf} proves the partition-function exponent.
Section~\ref{sec:main_proof} assembles the main theorem.
Section~\ref{sec:transfer} states the transfer theorem and
its consequences.
Section~\ref{sec:list} derives the list-guesswork exponent.
Section~\ref{sec:second_order} gives the second-order refinement.
Section~\ref{sec:universality} proves the
universality theorem and the $q$-ary exponent.
Section~\ref{sec:structured} applies the universality theorem to
Gallager's regular LDPC ensemble.
Section~\ref{sec:numerical} presents numerical validation.

%==================================================================
\section{System Model and Notation}
\label{sec:model}
%==================================================================

Let $n$ be the blocklength, $m=n(1-R)$ the number of parity checks,
and $k=nR$ the dimension, with $R\in(0,1)$.
The parity-check matrix $\Hm \in\Fb_2^{m\times n}$ is drawn uniformly
over all full-rank binary matrices of that size.
The code is $\Cb(\Hm)=\{c\in\Fb_2^n:\Hm c^T=\bm{0}\}$.
The noise vector is $e\sim\mathrm{Bernoulli}(p)^{\otimes n}$,
$p\in(0,\tfrac{1}{2})$, independent of $\Hm$, with
$P(e)=p^{\wH(e)}(1-p)^{n-\wH(e)}$.
Since $p<\tfrac{1}{2}$, $P(e)$ is \emph{strictly decreasing}
in the Hamming weight $\wH(e)$.
The syndrome is $\sigma=\Hm e^T\in\Fb_2^m$.
The coset of $e$ is
$\Nb(\Hm,\sigma)=\{e'\in\Fb_2^n:\Hm e'^T=\sigma\}$,
with $|\Nb(\Hm,\sigma)|=2^k$ when $\mathrm{rank}(\Hm)=m$.
The conditional distribution on the coset is
$Q_\sigma(e')=P(e')/Z_{\sigma} (1)$, where
\(
    Z_{\sigma} (\alpha)=
    \sum_{e'\in\Nb(\Hm,\sigma)}P(e')^\alpha,
     \alpha\in(0,1].
\)
Weight enumerator is
\(
    A_w(\Hm,\sigma)
    \;=\;
    \bigl|\bigl\{e'\in\Nb(\Hm,\sigma):\wH(e')=w\bigr\}\bigr|,
    \, w=0,\ldots,n.
\)
For $\alpha\in(0,1)$ the binary R\'{e}nyi entropy is
\(
    h_\alpha(p)
    \;=\;
    \frac{\log_2\!\bigl(p^\alpha+(1-p)^\alpha\bigr)}{1-\alpha},
\)
so $\log_2(p^\alpha+(1-p)^\alpha)=(1-\alpha)h_\alpha(p)$.
As $\alpha\to 1$, $h_\alpha(p)\to\Hb(p)$, \(\Hb(p)\) is binary entropy; for $\alpha<1$,
$h_\alpha(p)\ge\Hb(p)$.
At $\alpha=1/(1+\rho)$: $(1-\alpha)=\rho/(1+\rho)$.
Throughout, $\log$ and $\log_2$ denote natural and binary logarithms,
$\Pb$ probability, $\Eb$ expectation, and $o_{\Pb}(1)$ a sequence converging
to zero in probability.
Order the elements of $\Nb(\Hm,\sigma)$ by $P(\cdot)$ decreasingly.
The \emph{constrained guesswork} $G_{\mathrm{coset}}(e)$ is the
rank of $e$ in this ordering.
The \emph{constrained guesswork exponent} is
$\Lambda(\rho)=\lim_{n\to\infty}\frac{1}{n}\log_2\Eb[G_{\mathrm{coset}}^\rho]$,
$\rho>0$, whenever the limit exists.

\begin{example}
Running Example: The $(7,4,3)$ Hamming Code]
\label{ex:hamming_intro}
Throughout this paper we use the binary $(7,4,3)$ Hamming code as a
finite-length illustration. It has parameters $n=7$, $k=4$, $m=3$,
and rate $R=4/7 = 0.571$.
Its standard parity-check matrix is
\[
\Hm=\begin{pmatrix}
1&0&1&0&1&0&1\\
0&1&1&0&0&1&1\\
0&0&0&1&1&1&1
\end{pmatrix},
\]
with $\mathrm{rank}(\Hm)=3$.
Each syndrome $\sigma\in\Fb_2^3$ determines a coset of size
$|\Nb(\Hm,\sigma)|=2^k=16$.
For the \emph{random full-rank ensemble} with the same $(n,m)$,
the ensemble-average number of weight-$w$ coset representatives is
$\Eb[A_w]=\binom{7}{w}2^{-3}$.
These ensemble averages are used only as a comparison baseline;
they are not the exact weight counts of the fixed Hamming code.
\end{example}

\begin{definition}[Weight-enumerator growth rate]
\label{def:g_delta}
We say that the weight enumerator $A_w(\Hm,\sigma)$ has
\emph{growth rate} $g:[0,1]\to\mathbb{R}\cup\{-\infty\}$ if
\[
    \frac{1}{n}\log_2 A_{\lfloor n\delta\rfloor}(\Hm,\sigma)
    \;\xrightarrow{\;\Pb\;}\;
    g(\delta)
    \quad
    \text{for every } \delta\in(0,1).
\]
For the random full-rank ensemble of the present paper,
$g(\delta)=(\Hb(\delta)+R-1)^+$ (Theorem~\ref{thm:spectrum}).
\end{definition}

\begin{definition}[$L_n$-list constrained guesswork]
\label{def:list_guesswork}
Let $L_n\ge 1$ be an integer-valued sequence.
The \emph{$L_n$-list constrained guesswork} $G_{\mathrm{coset}}^{(L_n)}(e)$
is the rank of $e$ among all elements of $\Nb(\Hm,\sigma)$ when
ties within a Hamming-weight class are broken adversarially for
the first $L_n-1$ elements and uniformly for the $L_n$-th.
Equivalently,
\[
    G_{\mathrm{coset}}^{(L_n)}(e)
    \;=\;
    \left\lceil \frac{G_{\mathrm{coset}}(e)}{L_n} \right\rceil.
\]
\end{definition}

\begin{definition}[Partition-function variational functional]
\label{def:psi}
For any weight-enumerator growth rate $g:[0,1]\to\mathbb{R}\cup\{-\infty\}$
and log-probability slope
$\ell(\delta)=\delta\log_2 p+(1-\delta)\log_2(1-p)$,
define the \emph{$\alpha$-variational functional}
\begin{equation}
    \psi_\alpha(g)
    \;\triangleq\;
    \sup_{\delta\in[0,1]}
    \bigl[g(\delta)+\alpha\,\ell(\delta)\bigr],
    \quad \alpha\in(0,1].
    \label{eq:psi_def}
\end{equation}
For the binary full-rank ensemble, $g(\delta)=(\Hb(\delta)+R-1)^+$
and $\psi_\alpha(g)=(R-1)+(1-\alpha)h_\alpha(p)$ under
condition~\eqref{eq:subcritical}.
\end{definition}

\begin{definition}[$q$-ary R\'{e}nyi entropy and code ensemble]
\label{def:qary}
Let $q\ge 2$ be a prime power, $P:\mathbb{F}_q\to[0,1]$ a noise
distribution with $P(0)>P(a)$ for all $a\ne 0$.
The \emph{$q$-ary R\'{e}nyi entropy of order $\alpha$} is
\begin{equation}
    h_\alpha^{(q)}(P)
    \;\triangleq\;
    \frac{1}{1-\alpha}\log_2\!\left(\sum_{a\in\mathbb{F}_q}P(a)^\alpha\right),
    \quad \alpha\in(0,1).
    \label{eq:qary_renyi}
\end{equation}
A \emph{$q$-ary linear code ensemble} $\mathcal{E}_q$ consists of
uniformly random full-rank parity-check matrices
$\Hm\in\mathbb{F}_q^{m\times n}$, $m=n(1-R)$, with guesswork
$G_{\mathcal{E}_q}$ defined as the rank of $e$ in the coset
$\Nb_q(\Hm,\sigma)=\{e'\in\mathbb{F}_q^n:\Hm e'^T=\sigma\}$
ordered by decreasing $P^{\otimes n}$-probability.
\end{definition}

\paragraph{Connection to GRAND Decoding}

Consider the additive-noise channel
\(
y = x \oplus e,
\)
where \(x\in\Cb(\Hm)\) is a codeword, \(e\sim P^{\otimes n}\) is the
noise vector, and \(\sigma=\Hm y^T=\Hm e^T\) is the observed syndrome.
For a received word \(y\), syndrome-aided GRAND enumerates candidate
noise vectors \(e'\in\Fb_2^n\) in nonincreasing order of likelihood
\(P(e')\), while restricting the search to the syndrome coset
\(
\Nb(\Hm,\sigma)=\{e'\in\Fb_2^n:\Hm e'^T=\sigma\}.
\)
The decoder stops when it reaches the true noise realization \(e\).

\begin{theorem}[GRAND as constrained guesswork]
\label{thm:grand_connection}
For every fixed parity-check matrix \(\Hm\), syndrome \(\sigma\), and
noise realization \(e\in\Nb(\Hm,\sigma)\), the number of queries made by
syndrome-aided GRAND until successful decoding is exactly the constrained
guesswork random variable \(G_{\mathrm{coset}}(e)\).
Equivalently, under the additive-noise model,
\[
\text{GRAND query count} \;=\; G_{\mathrm{coset}}(e).
\]
Consequently, the moment exponent of GRAND satisfies
\[
\lim_{n\to\infty}\frac{1}{n}\log_2 \mathbb{E}\!\left[G_{\mathrm{coset}}^\rho\right]
=
\Lambda(\rho),
\]
whenever the limit exists.
\end{theorem}

\begin{proof}
Syndrome-aided GRAND searches only over the coset
\(\Nb(\Hm,\sigma)\) and orders candidate noises by decreasing
probability \(P(e')\). By definition, \(G_{\mathrm{coset}}(e)\) is the
rank of the true noise vector \(e\) in exactly this ordering. Hence the
decoder makes precisely \(G_{\mathrm{coset}}(e)\) queries before
stopping.
\end{proof}

\begin{remark}[Operational meaning]
Theorem~\ref{thm:grand_connection} identifies the paper's guesswork
analysis with the query complexity of syndrome-aided GRAND. In this
interpretation, the exponent \(\Lambda(\rho)\) quantifies the
asymptotic search cost of GRAND, and the reduction by \(1-R\) relative
to unconstrained guesswork is exactly the gain from syndrome
information.
\end{remark}

%==================================================================
\section{The Finite-Space Sandwich Bounds}
\label{sec:sandwich}
%==================================================================

\begin{theorem}[Finite-Space Guesswork Sandwich]
\label{thm:sandwich}
Let $\mathcal S$ be a finite set with cardinality $N$, and let $Q$ be a
probability mass function on $\mathcal S$. Order the elements of
$\mathcal S$ in nonincreasing order of probability under $Q$, breaking ties
arbitrarily, and define the guesswork random variable
\(
G(x)\triangleq \bigl|\{x'\in\mathcal S : Q(x')\ge Q(x)\}\bigr|.
\)
For any $\rho>0$, and
\(
Z_\sigma (\alpha)\triangleq \sum_{x\in\mathcal S} Q(x)^\alpha .
\)
Then
\begin{equation}
\frac{Z_\sigma (\alpha)^{1+\rho}}{\Harm_N^\rho}
\;\le\;
\mathbb E_Q\!\left[G(X)^\rho\right]
\;\le\;
Z_\sigma (\alpha)^{1+\rho},
\label{eq:sandwich_correct}
\end{equation}
where $\Harm_N=\sum_{j=1}^N j^{-1}$ is the $N$th harmonic number.
\end{theorem}

\begin{proof}
For the upper bound, note that for every $x\in\mathcal S$, and any $0<\alpha<1$,
\(
G(x)
=
\sum_{x'\in\mathcal S}\mathbf 1\{Q(x')\ge Q(x)\}
\le
\sum_{x'\in\mathcal S}\left(\frac{Q(x')}{Q(x)}\right)^\alpha
=
\frac{Z_\sigma (\alpha)}{Q(x)^\alpha},
\)
since $\mathbf 1\{t\ge 1\}\le t^\alpha$ for all $t\ge 0$.
Therefore,
\(
\mathbb E_Q[G(X)^\rho]
\le
Z_\sigma (\alpha)^\rho \sum_{x\in\mathcal S} Q(x)^{1-\alpha\rho}
=
Z_\sigma (\alpha)^{1+\rho},
\)
because $\alpha=1/(1+\rho)$ implies $1-\alpha\rho=\alpha$.

For the lower bound, order elements as $x_1,\ldots,x_N$ with
$q_j\triangleq Q(x_j)$ nonincreasing, so $G(x_j)=j$ and
$\mathbb{E}_Q[G^\rho]=\sum_j j^\rho q_j$.
Apply Hölder with exponents $1/\alpha$ and $1/(1-\alpha)$ to
$q_j^\alpha=(j^\rho q_j)^\alpha j^{-\alpha\rho}$ gives
\(
Z_\sigma(\alpha)\le\Bigl(\mathbb{E}_Q[G^\rho]\Bigr)^{\alpha}
\Harm_N^{1-\alpha},
\)
where $\alpha\rho/(1-\alpha)=1$ (since $\alpha=1/(1+\rho)$) collapses
the second sum to $\Harm_N^{1-\alpha}$.
Raising to $1/\alpha=1+\rho$ and rearranging gives
$\mathbb{E}_Q[G^\rho]\ge Z_\sigma(\alpha)^{1+\rho}/\Harm_N^\rho$.
In the coset setting $\frac{1}{n}\log_2\Harm_N^\rho=O(\log n/n)\to 0$,
so the harmonic penalty is asymptotically negligible.
\end{proof}

\begin{example}[Sandwich for the $(7,4,3)$ Hamming Code]
\label{ex:hamming_sandwich}
Apply Theorem~\ref{thm:sandwich} to the zero-syndrome coset
$\mathcal{S}=\Nb(\Hm,\bm{0})$ of the $(7,4,3)$ Hamming code,
which contains $N=2^4=16$ elements.
With $p=0.1$, $\rho=1$ ($\alpha=1/2$), the conditional
distribution $Q_{\bm{0}}$ assigns probability proportional to
$P(e')=(0.1)^{\wH(e')}(0.9)^{7-\wH(e')}$.
The code has weight distribution
$A_0=1$, $A_3=7$, $A_4=7$, $A_7=1$, so
\begin{align*}
Z_{\bm{0}}(1)&=(0.9)^7+7\cdot(0.1)^3(0.9)^4
+7\cdot(0.1)^4(0.9)^3+(0.1)^7\\
& = 0.4834.
\end{align*}
Likewise,
\(
Z_{\bm{0}}(1/2)=\sum_{e'\in\mathcal{S}}P(e')^{1/2}=0.956.
\)
and
$\varphi_{\bm{0}}=Z_{\bm{0}}(1/2)/Z_{\bm{0}}(1)^{1/2}
  = 1.375$.
The sandwich~\eqref{eq:sandwich_correct} with $Q=Q_{\bm{0}}$
(normalized) gives
\[
\frac{\varphi_{\bm{0}}^2}{\Harm_{16}}
\le
\Eb_{Q_{\bm{0}}}[G^1]
\le
\varphi_{\bm{0}}^2,
\]
i.e.\
$1.890/\Harm_{16}\le\Eb[G]\le1.890$,
where $\Harm_{16}=\sum_{j=1}^{16}j^{-1} = 3.381$.
Thus $0.559\le\Eb[G]\le1.890$.
confirming that even at blocklength $n=7$ the true mean guesswork
lies firmly between the two sandwich bounds.
As $n\to\infty$ the harmonic correction $\Harm_{2^k}^{\rho}$
contributes only $O(\log n/n)$ to the exponent,
which for $n=7$ is $\log_2(3.381)/7 = 0.249$ bits.
\end{example}

We now apply the sandwich to the coset.
Let $Q_\sigma$ be the conditional distribution on $\Nb(\Hm,\sigma)$, and let
\(
N_\sigma \triangleq |\Nb(\Hm,\sigma)|.
\)
For full-rank $\Hm$, we have $N_\sigma = 2^k$.

\begin{corollary}[Coset Sandwich]
\label{cor:coset_sandwich}
For any fixed $\Hm$, $\sigma$, and $\rho>0$, with $\alpha=1/(1+\rho)$,
\begin{align}
    \frac{1}{\Harm_{N_\sigma}^{\rho}}
    \left(\frac{Z_\sigma(\alpha)}{Z_\sigma(1)^\alpha}\right)^{1+\rho}
    \le
    \Eb\!\left[G_{\mathrm{coset}}^\rho \mid H,\sigma\right]
    \le
    \left(\frac{Z_\sigma(\alpha)}{Z_\sigma(1)^\alpha}\right)^{1+\rho}.
    \label{eq:coset_sandwich}
\end{align}
\end{corollary}
\begin{proof}
Apply Theorem~\ref{thm:sandwich} to $(\mathcal{S},Q)=(\Nb(\Hm,\sigma),Q_\sigma)$.
Since $Q_\sigma(e')=P(e')/Z_\sigma(1)$, we have
$\sum_{e'}Q_\sigma(e')^\alpha = Z_\sigma(\alpha)/Z_\sigma(1)^\alpha$,
which substituted into~\eqref{eq:sandwich_correct} gives~\eqref{eq:coset_sandwich}.
Taking the expectation over $(\Hm,e)$, $\mathrm{rank}(\Hm)=m$ holds with probability
$\ge 1-2^{-nR}$, we obtain
\begin{align}
    \frac{1}{\Harm_{2^k}^{\rho}}
    \Eb\bigl[\bigl(\tfrac{Z_\sigma(\alpha)}{Z_\sigma(1)^\alpha}\bigr)^{1+\rho}\bigr]
    \;\le\;
    \Eb\bigl[G_{\mathrm{coset}}^\rho\bigr]
    \;\le\;
    \Eb\bigl[\bigl(\tfrac{Z_\sigma(\alpha)}{Z_\sigma(1)^\alpha}\bigr)^{1+\rho}\bigr].
    \label{eq:global_sandwich}
\end{align}
\end{proof}

%==================================================================
\section{Uniform Weight-Spectrum Law}
\label{sec:spectrum}
%==================================================================

\subsection{Mean and Variance of the Weight Enumerator}

\begin{lemma}[Mean and variance of $A_w$]
\label{lem:Aw_moments}
For uniformly random full-rank $\Hm$ and any fixed syndrome $\sigma$, \( w=1,\ldots,n\), we have \(\mathrm{Var}(A_w) \;\le\; \Eb[A_w]\) and:
\begin{align}
    \Eb[A_w] \;=\; \binom{n}{w}2^{-m}, \label{eq:Aw_mean}
\end{align}
\end{lemma}

\begin{proof}
Write $A_w=\sum_{e:\wH(e)=w}\mathbf{1}[\Hm e^T=\sigma]$.
Under uniformly random $\Hm$, for any fixed $e\ne\bm{0}$,
$\Pb(\Hm e^T=\sigma)=2^{-m}$.
Hence $\Eb[A_w]=\binom{n}{w}2^{-m}$.
For the variance, write
$A_w^2 = \sum_{e,e' :\wH(e)=\wH(e')=w}
\mathbf{1}[\Hm e^T=\sigma]\mathbf{1}[\Hm e'^T=\sigma]$.
Separate the diagonal ($e=e'$) from the off-diagonal:
\begin{equation}
    \Eb[A_w^2]
    =
    \Eb[A_w]
    + \sum_{\substack{e\ne e'\\\wH(e)=\wH(e')=w}}\!\!\!
    \Pb(\Hm e^T=\sigma,\;\Hm e'^T=\sigma).
    \label{eq:Aw2}
\end{equation}
For distinct $e\ne e'$, the constraints $\Hm e^T=\sigma$ and
$\Hm e'^T=\sigma$ are equivalent to $\Hm e^T=\sigma$ and
$\Hm (e\oplus e')^T=\bm{0}$.
Since $e\oplus e'\ne\bm{0}$, the two constraints
$\Hm e^T=\sigma$ and $\Hm (e\oplus e')^T=\bm{0}$ are \emph{independent}
under the uniform distribution on $\Hm$:
conditioning on $\Hm e^T=\sigma$ (an event of probability $2^{-m}$)
does not change the marginal distribution of $\Hm (e\oplus e')^T$
as a random variable in $\Fb_2^m$, since $(e\oplus e')$ and $e$
are linearly independent over $\Fb_2$.
Therefore
$\Pb(\Hm e^T=\sigma,\,\Hm e'^T=\sigma)=2^{-m}\cdot 2^{-m}=2^{-2m}
=\Pb(\Hm e^T=\sigma)^2$.
Substituting into~\eqref{eq:Aw2}:
\begin{align*}
    \Eb[A_w^2]
    \;=\;
    \Eb[A_w]
    + \binom{n}{w}\!\left(\binom{n}{w}-1\right)2^{-2m} \\
    \;<\;
    \Eb[A_w]+\Eb[A_w]^2.
\end{align*}
Hence $\mathrm{Var}(A_w)=\Eb[A_w^2]-\Eb[A_w]^2<\Eb[A_w]$.
\end{proof}

\begin{example}[Weight-Enumerator Moments for the $(7,4,3)$ Hamming Code]
\label{ex:hamming_moments}
For the \emph{random full-rank ensemble} with $(n,m)=(7,3)$,
Lemma~\ref{lem:Aw_moments} gives
\[
\Eb[A_w]=\binom{7}{w}/8.
\]
Thus $\Eb[A_0]=1/8$, $\Eb[A_1]=7/8$, $\Eb[A_2]=21/8$, and
$\Eb[A_3]=35/8$.
Recall that the code has a weight distribution
$A_0=1,\, A_1=A_2=0,\, A_3=A_4=7,\, A_7=1$.
\end{example}

\subsection{Uniform Concentration via Chebyshev}

\begin{theorem}[Uniform Weight-Spectrum Law]
\label{thm:spectrum}
Let $\Hm \in\Fb_2^{m\times n}$ have i.i.d.\ $\mathrm{Bernoulli}(1/2)$
entries, and let $\sigma\in\Fb_2^m$ be an arbitrary fixed syndrome.
Fix $\eta>0$ and define
\[
    \Seta_\eta(R)
    \;=\;
    \{\delta\in[0,1]:\Hb(\delta)\ge 1-R+\eta\}.
\]
For $w\in\{0,\ldots,n\}$ with $w/n\in\Seta_\eta(R)$, call weight
class $w$ \emph{"$\varepsilon$-bad"} if either 1) $A_w(\Hm,\sigma)=0$ or 2)
\[
    \left|\frac{1}{n}\log_2 A_w(\Hm,\sigma)
    \;-\;\bigl(\Hb(w/n)+R-1\bigr)\right|
    \;>\;\varepsilon.
\]
Now, for every $\eps>0$ there exist constants $c,C>0$
depending only on $\eta,R,\eps$ such that for all sufficiently large $n$:
\begin{align}
    \Pb\!\left(
    \exists w\in\{0,\ldots,n\}
    \text{ with }
    \tfrac{w}{n}\in\Seta_\eta(R)
    \text{ that is $\varepsilon$-bad}
    \right)
    \le
    Ce^{-cn}.
    \label{eq:spectrum_conc}
\end{align}

The quantity $\Hb(\delta)+R-1$ is the correct exponent only
on $\Seta_\eta(R)$, where it is at least $\eta>0$ by definition.
Outside $\Seta_\eta(R)$, i.e.\ when $\Hb(\delta)<1-R$,
one has $\Eb[A_w]<1$ and $A_w=0$ with high probability,
so the correct unified exponent is $(\Hb(\delta)+R-1)^+$.

\end{theorem}

\begin{proof}
By definition of $\Seta_\eta(R)$, every $\delta$ in this set
satisfies $\Hb(\delta)\ge 1-R+\eta$, hence
\begin{equation}
    \Hb(\delta)+R-1 \;\ge\; \eta \;>\; 0.
    \label{eq:exponent_pos}
\end{equation}
Since $\Hm$ has i.i.d.\ $\mathrm{Bernoulli}(1/2)$ entries, for any
fixed $e\ne\bm{0}$ the vector $\Hm e^T$ is uniformly distributed on
$\Fb_2^m$, giving $\Pb(\Hm e^T=\sigma)=2^{-m}$ exactly.
By linearity of expectation, \eqref{eq:Aw_mean}, 
Stirling's approximation gives
$\frac{1}{n}\log_2\binom{n}{w}=\Hb(w/n)-O(\log n/n)$,
so
\begin{equation}
    \frac{1}{n}\log_2\Eb[A_w]
    \;=\;
    \Hb(w/n)+R-1-O\!\left(\tfrac{\log n}{n}\right).
    \label{eq:mean_exp}
\end{equation}
In particular, by~\eqref{eq:exponent_pos} and~\eqref{eq:mean_exp},
for $w/n\in\Seta_\eta(R)$ and all sufficiently large $n$,
\(
    \Eb[A_w] \;\ge\; 2^{n\eta/2}.
\)

Fix $w$ with $w/n\in\Seta_\eta(R)$.
By Markov's inequality and~\eqref{eq:Aw_mean}:
\begin{align}
    \Pb\!\left(A_w\ge 2^{n(\Hb(w/n)+R-1+\eps)}\right)
    \;\le\;
    \frac{\Eb[A_w]}{2^{n(\Hb(w/n)+R-1+\eps)}} \nonumber \\
    \;\stackrel{\eqref{eq:mean_exp}}{=}\;
    2^{-n\eps+O(\log n)}
    \;\le\;
    e^{-c_1 n}
\end{align}
for a constant $c_1=c_1(\eps)>0$ and all sufficiently large $n$.
By Lemma~\ref{lem:Aw_moments}, $\mathrm{Var}(A_w)\le\Eb[A_w]$.
Chebyshev's inequality gives
\(
    \Pb\!\left(A_w \le \tfrac{1}{2}\Eb[A_w]\right)
    \;\le\;
    \frac{4\,\mathrm{Var}(A_w)}{\Eb[A_w]^2}
    \;\le\;
    \frac{4}{\Eb[A_w]}
    \le\;
    4\cdot 2^{-n\eta/2}
    \;\le\;
    e^{-c_2 n}
\)
, since \(\Eb[A_w] \;\ge\; 2^{n\eta/2}\),
for some $c_2=c_2(\eta)>0$ and all sufficiently large $n$.
On the complementary event $A_w>\frac{1}{2}\Eb[A_w]$,
we have $A_w\ge 1$ (so $\log_2 A_w$ is well-defined) and
\begin{align}
    \frac{1}{n} & \log_2 A_w
    \;\ge\;
    \frac{1}{n}\log_2\Eb[A_w] - \frac{1}{n} \nonumber \\
    & \stackrel{\eqref{eq:mean_exp}}{=}\;
    \Hb(w/n)+R-1 - O\!\left(\tfrac{\log n}{n}\right),
\end{align}
which lies within $\eps$ of $\Hb(w/n)+R-1$ for all large $n$.
Thus both $A_w=0$ and the lower-deviation event
$\frac{1}{n}\log_2 A_w < \Hb(w/n)+R-1-\eps$
are contained in the event $\{A_w\le\frac{1}{2}\Eb[A_w]\}$,
whose probability is at most $e^{-c_2 n}$.

There are at most $n+1$ integers $w\in\{0,\ldots,n\}$
with $w/n\in\Seta_\eta(R)$.
A union bound over these gives
\[
\Pb\!\left(\exists\, w \text{ with } \tfrac{w}{n}\in\Seta_\eta(R) 
\text{ that is $\varepsilon$-bad}\right)
\;\le\;
(n+1)\bigl(e^{-c_1 n}+e^{-c_2 n}\bigr)
\;\le\;
C\,e^{-cn},
\]
where $c=\min(c_1,c_2)/2$ and $C$ is an absolute constant
absorbing the polynomial $n+1$.

To pass from the integer grid $\{w/n\}$ to the continuous
set $\Seta_\eta(R)$: the set $\Seta_\eta(R)$ is a closed
interval $[\delta_{\min},\delta_{\max}]\subset(0,1)$ with
endpoints depending only on $\eta$ and $R$.
On this interval, $\Hb'(\delta)=\log_2\frac{1-\delta}{\delta}$
is bounded in absolute value by a constant $L=L(\eta,R)<\infty$
independent of $n$.
Therefore, for every $\delta\in\Seta_\eta(R)$,
\(
    \left|\Hb\!\left(\tfrac{\lfloor\delta n\rfloor}{n}\right)
    -\Hb(\delta)\right|
    \;\le\;
    \frac{L(\eta,R)}{n}
    \;=\;
    O(1/n),
\)
which is absorbed into $\eps$ for large $n$, justifying the
sup over $\delta\in\Seta_\eta(R)$ with $\lfloor\delta n\rfloor$
in place of $\delta n$.

For the uniform full-rank ensemble,
$\Pb(\mathrm{rank}(\Hm)<m)\le 2^{-nR}$.
For any event $E$,
\(
    \Pb(E \mid \mathrm{rank}(\Hm)=m)
    \;\le\;
    \frac{\Pb(E)}{1-2^{-nR}}.
\)
Since $(1-2^{-nR})^{-1}$ is a subexponential prefactor, the
exponential decay rate in~\eqref{eq:spectrum_conc} is unchanged.
\end{proof}

\begin{example}[Spectrum Law for the $(7,4,3)$ Hamming Code]
\label{ex:hamming_spectrum}
For the $(7,4,3)$ Hamming code, $R=4/7$ and $1-R=3/7 = 0.429$.
The high-entropy region at $\eta=0.05$ is
$\Seta_{0.05}(R)=\{\delta:\Hb(\delta)\ge3/7+0.05 = 0.479\}$,
which corresponds roughly to $\delta\in[0.092,0.908]$.
The weight fractions $\delta=w/7$ for $w=1,\ldots,6$ all fall inside
$\Seta_\eta(R)$ (since $\Hb(1/7) = 0.592>0.479$).
The asymptotic exponent for weight class $w$ is
$\Hb(w/7)+4/7-1=\Hb(w/7)-3/7$.
For $w=3$ ($\delta=3/7$): $\Hb(3/7) = 0.985$, giving
exponent $ = 0.556$.
For $w=1$ ($\delta=1/7$): the ensemble mean is
$\Eb[A_1]=\binom{7}{1}/8=0.875$.
Because $n=7$ is very small, this example should be read only as a
qualitative sanity check for the asymptotic theorem.
\end{example}

%==================================================================
\section{Partition-Function Exponent Theorem}
\label{sec:pf}
%==================================================================

Grouping $Z_\sigma(\alpha)$ by Hamming weight gives
\[
Z_\sigma(\alpha)=(1-p)^{\alpha n}\sum_w A_w\bigl(\tfrac{p}{1-p}\bigr)^{\alpha w},
\]
a discrete Laplace transform with a negative exponent, $\alpha\log(p/(1-p))<0$, since $p<\tfrac{1}{2}$.

\subsection{The Maximizing Weight Class}

Before stating the theorem we identify the weight fraction
$\delta^*$ that dominates the partition function.
Recall $f:[0,1]\to \Rb$ defined by
\[
    f(\delta)=
    \Hb(\delta)
    +\alpha\delta\log_2 p
    +\alpha(1-\delta)\log_2(1-p).
\]
Since $f''(\delta)=-1/(\delta(1-\delta)\ln 2)<0$, $f$ is strictly
concave on $(0,1)$ and has a unique maximizer.
Setting $f'(\delta^*)=0$ gives
\[
    \delta^*
    =
    \frac{p^\alpha}{p^\alpha+(1-p)^\alpha}
    \in(0,\tfrac{1}{2}),
\]
where $\delta^*<\tfrac{1}{2}$ follows from $p<\tfrac{1}{2}$.

\begin{lemma}[Value at the maximizer]
\label{lem:f_max}
\[
\displaystyle f(\delta^*)
=\log_2\!\bigl(p^\alpha+(1-p)^\alpha\bigr)
=(1-\alpha)h_\alpha(p).\]
\end{lemma}
\begin{proof}
With \(\beta=p^\alpha+(1-p)^\alpha\), we have \(\delta^*=p^\alpha/\beta\) and
\(1-\delta^*=(1-p)^\alpha/\beta\). Substituting into \(f\) gives
\(f(\delta^*)=\log_2 \beta\), hence the claim.
\end{proof}

\begin{example}[Maximizing Weight Fraction for the $(7,4,3)$ Code]
\label{ex:hamming_delta_star}
With $p=0.1$, $\rho=1$, $\alpha=1/2$:
\[
\delta^*=\frac{(0.1)^{1/2}}{(0.1)^{1/2}+(0.9)^{1/2}}
=\frac{\sqrt{0.1}}{\sqrt{0.1}+\sqrt{0.9}}
 = 0.2500.
\]
Thus the partition function is centered near weight fraction
$\delta^* = 0.25$.
The value
\[
f(\delta^*)=\log_2(0.1^{1/2}+0.9^{1/2})
 = 0.3390\text{ bits}
\]
matches $(1-\alpha)h_{1/2}(0.1)=0.3390$ 
bits since \(h_{1/2}(0.1) =\frac{\log_2(0.1^{1/2}+0.9^{1/2})}{1/2} = 0.6781\) bits.
For the fixed $(7,4,3)$ code, this saddlepoint is only an asymptotic guide:
it does not imply that the weight-1 class must be occupied.
\end{example}

\begin{lemma}[$\delta^*$ lies in the high-entropy region]
\label{lem:delta_star_in_S}
Suppose
\begin{equation}
    \Hb(\delta^*) \;>\; 1-R
    \quad\text{and}\quad
    \Hb(p) \;>\; 1-R.
    \label{eq:subcritical}
\end{equation}
Then there exists $\eta_0=\eta_0(p,\alpha,R)>0$ such that
$\delta^*\in\Seta_\eta(R)$ for all $\eta\le\eta_0$.
\end{lemma}
\begin{proof}
Set $\eta_0=\tfrac{1}{2}(\Hb(\delta^*)-(1-R))>0$.
Then $\Hb(\delta^*)\ge 1-R+\eta_0$, so $\delta^*\in\Seta_{\eta_0}(R)$,
and hence $\delta^*\in\Seta_\eta(R)$ for all $\eta\le\eta_0$.
\end{proof}

\begin{example}[Subcriticality Check for the $(7,4,3)$ Code]
\label{ex:hamming_subcritical}
With $\delta^* = 0.25$, $p=0.1$, and $R=4/7$:
\(
\Hb(\delta^*)=\Hb(0.25)=-(0.25\log_2 0.25+0.75\log_2 0.75) = 0.8113\text{ bits},
1-R=3/7 = 0.4286,
\Hb(p)=\Hb(0.1) = 0.4690\text{ bits}.
\)
Both conditions in~\eqref{eq:subcritical} are satisfied:
$0.8113>0.4286$ and $0.4690>0.4286$.
Thus $\eta_0=\frac{1}{2}(0.8113-0.4286) = 0.191$,
and the dominant saddlepoint $\delta^*=0.25$ is well inside
the high-entropy region for this code, confirming that the
asymptotic analysis applies without modification.
\end{example}

\subsection{Main Concentration Result}

\begin{theorem}[Partition-Function Exponent]
\label{thm:pf_exp}
Let $\Hm \in\Fb_2^{m\times n}$ have i.i.d.\ $\mathrm{Bernoulli}(1/2)$
entries, let $\sigma\in\Fb_2^m$ be a fixed syndrome, and assume
condition~\eqref{eq:subcritical}.
Then for every $\alpha\in(0,1]$:
\begin{align}
    \frac{1}{n}\log_2 Z_\sigma(\alpha)
    &\;=\;
    R-1+(1-\alpha)h_\alpha(p)\;+\;o_{\Pb}(1),
    \label{eq:pf_exp_alpha}
\end{align}

\end{theorem}

\begin{proof}
Choose $\eta=\eta_0/2$ (Lemma~\ref{lem:delta_star_in_S}) so
$\delta^*\in\Seta_\eta(R)$, and let $\mathcal{E}_n$ denote the
good event of Theorem~\ref{thm:spectrum} at this $\eta,\varepsilon$,
so $\Pb(\mathcal{E}_n^c)\le Ce^{-cn}$ and on $\mathcal{E}_n$, \(w/n\in\Seta_\eta(R)\) :
\begin{align}
    2^{n(\Hb(w/n)+R-1-\varepsilon)}\le A_w\le
    2^{n(\Hb(w/n)+R-1+\varepsilon)}.
    \label{eq:spectrum_event}
\end{align}

We have: 
\(Z_{\rm in}=\sum_{\substack{w=0, w/n\,\in\,\Seta_\eta(R)}}^n
        A_w\,p^{\alpha w}(1-p)^{\alpha(n-w)}\),
 \(Z_{\rm out}=\sum_{\substack{w=0, w/n\,\notin\,\Seta_\eta(R)}}^n
        A_w\,p^{\alpha w}(1-p)^{\alpha(n-w)}\), which gives
\(Z_\sigma(\alpha)=Z_{\rm in}+Z_{\rm out}\).

\noindent\textbf{Step 1: $Z_{\rm out}=0$ w.h.p.}
For $w/n\notin\Seta_\eta(R)$, Lemma~\ref{lem:Aw_moments} gives
\(
    \Eb[A_w]
    \;=\;
    \binom{n}{w}2^{-m}
    \;\le\;
    2^{n(\Hb(w/n)+R-1)}
    \;\le\;
    2^{-n\eta}.
\)
By Markov's inequality, $\Pb(A_w\ge 1)\le\Eb[A_w]\le 2^{-n\eta}$.
Taking a union bound over the at most $n+1$ such weight classes:
\(
    \Pb\!\left(\exists\,w\text{ with }w/n\notin\Seta_\eta(R)
    \text{ s.t.\ }A_w\ge 1\right)
    \;\le\;
    (n+1)2^{-n\eta}
    \;\le\;
    C'e^{-c'n}.
\)
Denote this event
$\mathcal{F}_n=\{\forall\, w \text{ with }
w/n\notin\Seta_\eta(R),\; A_w=0\}$,
so $\Pb(\mathcal{F}_n^c)\le C'e^{-c'n}$.
 
\noindent\textbf{Step 2: Upper bound on $Z_{\rm in}$ on $\mathcal{E}_n$.}
On $\mathcal{E}_n$, each term in $Z_{\rm in}$ satisfies
$A_w\le 2^{n(\Hb(w/n)+R-1+\varepsilon)}$, so:
\[
    Z_{\rm in}
    \;\le\;
    \sum_{\substack{w:w/n\,\in\,\Seta_\eta(R)}}
    2^{n(\Hb(w/n)+R-1+\varepsilon)}\,p^{\alpha w}(1-p)^{\alpha(n-w)}
    \;=\;
    \sum_{\substack{w:w/n\,\in\,\Seta_\eta(R)}}
    2^{n(f(w/n)+R-1+\varepsilon)} 
    \;\le\;
    (n+1)\cdot
    2^{n\!\left(\max_{\delta\in\Seta_\eta(R)}f(\delta)+R-1+\varepsilon\right)}.
\]
By Lemma~\ref{lem:f_max},
$\max_{\delta\in\Seta_\eta(R)}f(\delta)=f(\delta^*)=(1-\alpha)h_\alpha(p)$.
Using $(n+1)=2^{O(\log n)}$:
\begin{equation}
    \frac{1}{n}\log_2 Z_{\rm in}
    \;\le\;
    (R-1)+(1-\alpha)h_\alpha(p)+\varepsilon+O\!\left(\tfrac{\log n}{n}\right).
    \label{eq:upper_final}
\end{equation}

\noindent\textbf{Step 3: Lower bound on $Z_{\rm in}$ on $\mathcal{E}_n$.}
Let $w^*=\lfloor\delta^* n\rfloor$.
Since $\delta^*\in\Seta_\eta(R)$, the bound~\eqref{eq:spectrum_event}
applies to $w^*$ on $\mathcal{E}_n$:
\(
    Z_{\rm in}
    \;\ge\;
    A_{w^*}\,p^{\alpha w^*}(1-p)^{\alpha(n-w^*)}
    \;\ge\;
    2^{n(\Hb(w^*/n)+R-1-\varepsilon)}\cdot
    2^{n\alpha(w^*/n\cdot\log_2 p\,+\,(1-w^*/n)\cdot\log_2(1-p))}
    \;=\;
    2^{n(f(w^*/n)+R-1-\varepsilon)}.
\)
Since $f$ is Lipschitz on $\Seta_\eta(R)$ with constant
$L=L(\eta,R)<\infty$ and $|w^*/n-\delta^*|\le 1/n$:
\(
    f(w^*/n)\;\ge\; f(\delta^*)-L/n
    \;=\;(1-\alpha)h_\alpha(p)-O(1/n).
\)
Hence:
\begin{equation}
    \frac{1}{n}\log_2 Z_{\rm in}
    \;\ge\;
    (R-1)+(1-\alpha)h_\alpha(p)-\varepsilon-O\!\left(\tfrac{1}{n}\right).
    \label{eq:lower_final}
\end{equation}

\noindent\textbf{Conclusion.}
Let $\mathcal{G}_n=\mathcal{E}_n\cap\mathcal{F}_n$.
Then $\Pb(\mathcal{G}_n^c)\le C''e^{-c''n}$ for constants
depending on $\eta,R,\varepsilon$.
On $\mathcal{G}_n$, $Z_\sigma(\alpha)=Z_{\rm in}$ and
combining~\eqref{eq:upper_final}--\eqref{eq:lower_final} gives
\(
    \left|
    \frac{1}{n}\log_2 Z_\sigma(\alpha)
    -(R-1)-(1-\alpha)h_\alpha(p)
    \right|
    \;\le\;
    \varepsilon+O\!\left(\tfrac{\log n}{n}\right).
\)
Since $\varepsilon>0$ was arbitrary, \eqref{eq:pf_exp_alpha} follows.
\end{proof}

%==================================================================
\section{Proof of the Main Theorem}
\label{sec:main_proof}
%==================================================================

\begin{theorem}[Exact Guesswork Exponent]
\label{thm:main}
Let $\Hm \in\Fb_2^{m\times n}$ have i.i.d.\ $\mathrm{Bernoulli}(1/2)$
entries with $m=n(1-R)$, and let
$e\sim\mathrm{Bernoulli}(p)^{\otimes n}$, $p\in(0,\tfrac{1}{2})$,
be independent of $\Hm$.
Assume condition~\eqref{eq:subcritical}.
Then for every $\rho>0$, writing $\alpha=1/(1+\rho)$:
\begin{equation}
    \lim_{n\to\infty}
    \frac{1}{n}\log_2\Eb\!\left[G_{\mathrm{coset}}^{\rho}\right]
    \;=\;
    \rho\,h_{\alpha}(p)\;+\;\rho(R-1).
    \label{eq:main}
\end{equation}
\end{theorem}

\begin{proof}
Set $\alpha=1/(1+\rho)$ and
$\varphi_\sigma = Z_\sigma(\alpha)/Z_\sigma(1)^\alpha$.

\noindent\textbf{Step 1: Exponent of $\varphi_\sigma$.}
From Theorem~\ref{thm:pf_exp},
$\frac{1}{n}\log_2 Z_\sigma(\alpha) = (R-1)+(1-\alpha)h_\alpha(p)+o_\Pb(1)$
and $\frac{1}{n}\log_2 Z_\sigma(1) = R-1+o_\Pb(1)$.
(the second part of condition~\eqref{eq:subcritical}
ensures $p=\delta^*|_{\alpha=1}\in\Seta_\eta(R)$,
so Theorem~\ref{thm:pf_exp} applies at $\alpha=1$.)
Therefore:
\(
    \frac{1}{n}\log_2\varphi_\sigma
    \;=\;
    \frac{1}{n}\log_2 Z_\sigma(\alpha)
    -\alpha\cdot\frac{1}{n}\log_2 Z_\sigma(1)
    \;=\;
    \bigl[(R-1)+(1-\alpha)h_\alpha(p)\bigr]
    -\alpha(R-1)
    +o_{\Pb}(1)
    \;=\;
    (1-\alpha)\bigl[h_\alpha(p)+(R-1)\bigr]+o_{\Pb}(1).
\)
Multiplying by $(1+\rho)$ and using $(1+\rho)(1-\alpha)=\rho$:
\begin{equation}
    \frac{1}{n}\log_2\varphi_\sigma^{1+\rho}
    \;=\;
    \rho\bigl[h_\alpha(p)+(R-1)\bigr]+o_{\Pb}(1).
    \label{eq:phi_rate}
\end{equation}

\noindent\textbf{Step 2: From $o_\Pb(1)$ to expectation.}
We need
\begin{equation}
    \frac{1}{n}\log_2\Eb\!\left[\varphi_\sigma^{1+\rho}\right]
    \;=\;
    \rho h_\alpha(p)+\rho(R-1)+o(1).
    \label{eq:phi_moment}
\end{equation}
Let $\Lambda=\rho h_\alpha(p)+\rho(R-1)$ and
$W=\varphi_\sigma^{1+\rho}/2^{n\Lambda}$,
so $\frac{1}{n}\log_2 W=o_\Pb(1)$ by~\eqref{eq:phi_rate}.
Let $\mathcal{G}_n$ be the good event of Theorem~\ref{thm:pf_exp}
at parameter $\varepsilon$; on $\mathcal{G}_n$,
$2^{-n(1+\rho)\varepsilon}\le W\le 2^{n(1+\rho)\varepsilon}$
and $\Pb(\mathcal{G}_n^c)\le C''e^{-c''n}$.
Since $Z_\sigma(\alpha)\le(p^\alpha+(1-p)^\alpha)^n$ and
$Z_\sigma(1)\ge p^n$, we have $W\le 2^{nD}$ a.s.\ for a finite
constant $D$,
where $D=(1+\rho)\bigl[(1-\alpha)h_\alpha(p)
+\alpha(1-R)+\alpha\log_2(1/p)\bigr]<\infty$.
Splitting on $\mathcal{G}_n$, \(\Eb[W]\le\Eb[W\mathbf{1}_{\mathcal{G}_n}] +\Eb[W\mathbf{1}_{\mathcal{G}_n^c}]\) gives
\(
2^{-n(1+\rho)\varepsilon}(1-C''e^{-c''n})
\;\le\;\Eb[W]\;\le\;
2^{n(1+\rho)\varepsilon}+2^{nD}\cdot C''e^{-c''n}.
\)
As $\varepsilon\to 0$ and $n\to\infty$, both bounds tend to $1$,
giving $\frac{1}{n}\log_2\Eb[W]\to 0$, i.e.,~\eqref{eq:phi_moment}.

\noindent\textbf{Step 3: Applying the sandwich.}
From~\eqref{eq:global_sandwich}:
\begin{equation}
    \frac{\Eb[\varphi_\sigma^{1+\rho}]}{\Harm_{2^k}^\rho}
    \;\le\;
    \Eb\!\left[G_{\mathrm{coset}}^\rho\right]
    \;\le\;
    \Eb\!\left[\varphi_\sigma^{1+\rho}\right],
    \label{eq:global_sandwich_used}
\end{equation}
where $\Harm_{2^k}=\sum_{j=1}^{2^k}j^{-1}\le 1+k\ln 2=1+nR\ln 2$
is the $2^k$-th harmonic number.
Taking $\frac{1}{n}\log_2$ in~\eqref{eq:global_sandwich_used}
and using~\eqref{eq:phi_moment}:
\begin{align}
    \rho h_\alpha(p)+\rho(R-1)- & O\!\left(\tfrac{\log n}{n}\right)+o(1)
    \;\le\;
    \frac{1}{n}\log_2\Eb[G_{\rm coset}^\rho]
    \notag\\
    &\;\le\;
    \rho h_\alpha(p)+\rho(R-1)+o(1).
    \label{eq:final_sandwich}
\end{align}
\end{proof}

\begin{example}[Guesswork Exponent for the $(7,4,3)$ Code]
\label{ex:hamming_main}
Theorem~\ref{thm:main} gives the asymptotic exponent for the
$(7,4,3)$ Hamming code with $p=0.1$, $\rho=1$, $\alpha=1/2$:
\begin{align*}
\Lambda(1)
=\rho\,h_{1/2}(0.1)+\rho(R-1)
 = 0.2495\text{ bits}.
\end{align*}
This means that for large blocklengths at rate $4/7$,
$\Eb[G_{\mathrm{coset}}^1] =  2^{0.2495n}$.
The unconstrained Ar{\i}kan--Merhav exponent would be
$\Lambda_{\rm unc}(1)=h_{1/2}(0.1) = 0.6781$ bits, so the syndrome
constraint reduces only the \emph{first-order exponent} by $1-R=3/7$ bits.
At $n=7$ this is only a rough heuristic; it should not be read as a
finite-length prediction.
\end{example}

\begin{corollary}[Exact Complexity Reduction]
\label{cor:reduction}
Write $\Lambda_{\mathrm{unc}}(\rho)=\rho h_\alpha(p)$ for the
unconstrained Ar{\i}kan--Merhav exponent.  Then:
\(
    \Lambda(\rho)
    \;=\;
    \Lambda_{\mathrm{unc}}(\rho)\;-\;\rho(1-R).
\)
Each of the $n(1-R)$ parity-check constraints reduces the guesswork
exponent by exactly $\rho/n$, and the total reduction $\rho(1-R)$
is tight: equality is achieved in the limit.
\end{corollary}

\begin{lemma}[Query Budget Threshold]
\label{cor:budget}
Let $B_n=2^{n\gamma}$ be a decoder query budget.
Under the conditions of Theorem~\ref{thm:main}, for every $\rho>0$:
\(
    \Pb\!\left(G_{\mathrm{coset}}>B_n\right)
    \;\le\;
    2^{n[\Lambda(\rho)-\rho\gamma]+o(n)}.
\)
In particular, if $\gamma>h_\alpha(p)+(R-1)$, then
$\Pb(G_{\mathrm{coset}}>B_n)\to 0$ exponentially in $n$.
\end{lemma}

\begin{proof}
Markov's inequality applied to $G_{\mathrm{coset}}^\rho$ gives
\(
    \Pb(G_{\mathrm{coset}}>B_n)
    \;\le\;
    \frac{\Eb[G_{\mathrm{coset}}^\rho]}{B_n^\rho}
    \;=\;
    2^{\frac{1}{n}\log_2\Eb[G_{\mathrm{coset}}^\rho]\cdot n
       -n\rho\gamma}.
\)
Theorem~\ref{thm:main} gives
$\frac{1}{n}\log_2\Eb[G_{\mathrm{coset}}^\rho]
=\rho h_\alpha(p)+\rho(R-1)+o(1)$,
so the exponent equals $n[\Lambda(\rho)-\rho\gamma]+o(n)$.
When $\gamma>h_\alpha(p)+(R-1)$, the coefficient
$\Lambda(\rho)-\rho\gamma<0$ and the bound decays
exponentially.
\end{proof}

\begin{example}[Query Budget for the $(7,4,3)$ Hamming Code]
\label{ex:hamming_budget}
Let $B_n=2^{n\gamma}$ be the decoder budget.
From Example~\ref{ex:hamming_main}, $\Lambda(1) = 0.2495$ bits.
The threshold is $\gamma^*=h_{1/2}(0.1)+(R-1) = 0.2495$.
\begin{itemize}
\item If $\gamma=0.30>\gamma^*$, then
\[
\Pb(G_{\mathrm{coset}}>2^{0.3n})
\le 2^{n(\gamma^*-0.30)+o(n)}
=2^{-0.0505n+o(n)},
\]
which decays exponentially in $n$.
\item If $\gamma=\gamma^*$, the bound is asymptotically neutral
and therefore not informative.
\item For $n=7$, this bound is only a loose Markov estimate and should
not be interpreted as an exact success probability.
\end{itemize}
\end{example}

\begin{remark}[Communication-theoretic interpretation]
The threshold $\gamma>h_\alpha(p)+(R-1)$ is the syndrome-aware analogue of the unconstrained GRAND budget threshold $\gamma>h_\alpha(p)$. Thus, observing the syndrome reduces the required query exponent by $1-R$.
\end{remark}

\begin{example}[Communication-Theoretic Interpretation for $(7,4,3)$]
\label{ex:hamming_comm}
For the $(7,4,3)$ Hamming code, the syndrome reduces the
required budget threshold from
$h_{1/2}(0.1) = 0.6781$ bits (unconstrained GRAND)
to $h_{1/2}(0.1)+(R-1) = 0.2495$ bits (syndrome-aware GRAND).
In concrete terms, to guarantee $\Pb(G_{\mathrm{coset}}>B_n)\to0$,
unconstrained GRAND needs budget $B_n=2^{0.6781n}$ while
syndrome-aware GRAND needs only $B_n=2^{0.2495n}$: at $n=70$,
this is $2^{47.5} = 1.9\times10^{14}$ vs.\
$2^{17.5} = 1.9\times10^5$, roughly a $10^9$-fold saving
from the 3 syndrome bits per block.
\end{example}

%==================================================================
% NEW SECTION: Transfer Theorem (Session 1 blue)
%==================================================================

\section{Transfer Theorem: From Spectrum to Guesswork}
\label{sec:transfer}

The proof of Theorem~\ref{thm:pf_exp} rested on a single structural
fact: the weight enumerator $A_w(\Hm,\sigma)$ concentrates at
exponential rate $g(\delta)\triangleq\Hb(\delta)+R-1$ inside the high-entropy
region.  We now abstract this observation into a \emph{transfer
theorem} that converts any spectrum growth rate $g(\delta)$ into a
partition-function exponent and, via the sandwich of
Theorem~\ref{thm:sandwich}, into a guesswork exponent.
Equation~\eqref{eq:spectrum_boxed} states the concentration behavior of the random full-rank ensemble in the notation of Definition~\ref{def:g_delta}.

\begin{theorem}[Spectrum Concentration]
\label{thm:spectrum_boxed}
Under the conditions of Theorem~\ref{thm:spectrum}, for every
$\delta\in(0,1)$ with $\Hb(\delta)>1-R$:
\begin{equation}
    \frac{1}{n}\log_2 A_{\lfloor n\delta\rfloor}(\Hm,\sigma)
    \;=\;
    g(\delta)\;+\;o_{\Pb}(1).
\label{eq:spectrum_boxed}
\end{equation}
Outside this region ($\Hb(\delta)\le 1-R$), $A_{\lfloor n\delta\rfloor}=0$
with high probability, so the correct unified growth rate is
$g(\delta)=(\Hb(\delta)+R-1)^+$.
\end{theorem}

\begin{proof}
This is an immediate restatement of Theorem~\ref{thm:spectrum}
in the notation of Definition~\ref{def:g_delta}.
\end{proof}

\subsection{Transfer Theorem}

Let $\ell:[0,1]\to\mathbb{R}$ be the log-probability slope for
Bernoulli$(p)$ noise:
\begin{equation}
    \ell(\delta)
    \;=\;
    \delta\log_2 p\;+\;(1-\delta)\log_2(1-p),
    \quad \delta\in[0,1].
    \label{eq:ell_def}
\end{equation}
Note that $\ell(\delta)<0$ for all $\delta\in[0,1)$ and
$f(\delta)=g(\delta)+\alpha\,\ell(\delta)$ in the notation of
Section~\ref{sec:pf}.

\begin{theorem}[Transfer Theorem]
\label{thm:transfer}
Let $g:[0,1]\to\mathbb{R}\cup\{-\infty\}$ be any function such that
\begin{equation}
    \frac{1}{n}\log_2 A_{\lfloor n\delta\rfloor}(\Hm,\sigma)
    \;\xrightarrow{\;\Pb\;}\;
    g(\delta)
    \quad\text{for every }\delta\in(0,1),
    \label{eq:g_assumption}
\end{equation}
and assume that $g$ is continuous on its effective domain
$\{\delta:g(\delta)>-\infty\}$ and that the supremum
\begin{equation}
    \psi_\alpha(g)
    \;\triangleq\;
    \sup_{\delta\in[0,1]}
    \bigl[g(\delta)+\alpha\,\ell(\delta)\bigr]
    \label{eq:Gamma_def}
\end{equation}
is achieved at a unique interior point $\delta^*(\alpha)\in(0,1)$.
Then:
\begin{equation}
    \frac{1}{n}\log_2 Z_\sigma(\alpha)
    \;=\;
    \psi_\alpha(g)
    \;+\;o_{\Pb}(1).
\label{eq:transfer_boxed}
\end{equation}
Consequently, the constrained guesswork exponent satisfies
\begin{equation}
    \Lambda(\rho)
    \;=\;
    (1+\rho)\,\psi_{1/(1+\rho)}(g)
    \;-\;
    \rho\,\psi_1(g),
    \label{eq:transfer_exponent}
\end{equation}
where $\alpha=1/(1+\rho)$ and both suprema are achieved.
\end{theorem}

\begin{proof}
The partition function decomposes by Hamming weight as
\(
    Z_\sigma(\alpha)
    \;=\;
    \sum_{w=0}^n
    A_w(\Hm,\sigma)\,p^{\alpha w}(1-p)^{\alpha(n-w)}
    \;=\;
    \sum_{w=0}^n
    A_w(\Hm,\sigma)\cdot 2^{n\alpha\ell(w/n)}.
\)
Under assumption~\eqref{eq:g_assumption}, $A_{\lfloor n\delta\rfloor}
=2^{n(g(\delta)+o_\Pb(1))}$ for each $\delta$ in the effective
domain, so each term contributes $2^{n(g(w/n)+\alpha\ell(w/n)+o_\Pb(1))}$.
The sum is dominated exponentially by the term at $w^*=\lfloor\delta^*
n\rfloor$, since the number of summands is at most $n+1=2^{O(\log n)}$:
\begin{align}
\frac{1}{n}\log_2 Z_\sigma(\alpha)
&\le
\max_{\delta\in[0,1]}
\bigl[g(\delta)+\alpha\ell(\delta)\bigr]
+o_\Pb(1),
\\
\frac{1}{n}\log_2 Z_\sigma(\alpha)
&\ge
g(\delta^*)+\alpha\ell(\delta^*)
+o_\Pb(1),
\end{align}
Since $g$ is continuous at $\delta^*$ and $|w^*/n-\delta^*|\le 1/n$,
the Lipschitz error is $O(1/n)$, and~\eqref{eq:transfer_boxed} follows.
Equation~\eqref{eq:transfer_exponent} then follows from
the same algebra as the proof of Theorem~\ref{thm:main}
(Steps~1--3), replacing $(1-\alpha)h_\alpha(p)$ by $\psi_\alpha(g)-\alpha\psi_1(g)$.
\end{proof}

\begin{remark}[Random full-rank ensemble as a special case]
For the random full-rank ensemble, $g(\delta)=(\Hb(\delta)+R-1)^+$
(Theorem~\ref{thm:spectrum_boxed}).
The saddlepoint $\delta^*(\alpha)=p^\alpha/(p^\alpha+(1-p)^\alpha)$
gives $\psi_\alpha(g)=(R-1)+(1-\alpha)h_\alpha(p)$
and $\psi_1(g)=R-1$,
so~\eqref{eq:transfer_exponent} yields
$\Lambda(\rho)=\rho h_\alpha(p)+\rho(R-1)$,
recovering Theorem~\ref{thm:main} exactly.
\end{remark}

\begin{example}[Transfer Theorem Applied to the $(7,4,3)$ Code]
\label{ex:hamming_transfer}
The $(7,4,3)$ Hamming code has $g(\delta)=(\Hb(\delta)+4/7-1)^+
=(\Hb(\delta)-3/7)^+$.
With $p=0.1$, $\alpha=1/2$, the log-probability slope is
$\ell(\delta)=\delta\log_2(0.1)+(1-\delta)\log_2(0.9)
 = -3.322\delta-0.152(1-\delta)$,
and the variational functional is
\[
\psi_{1/2}(g)=\sup_\delta\bigl[(\Hb(\delta)-3/7)^+
+\tfrac{1}{2}\ell(\delta)\bigr].
\]
The maximizer is $\delta^*=\sqrt{0.1}/(\sqrt{0.1}+\sqrt{0.9})
 = 0.250$ (Example~\ref{ex:hamming_delta_star}), and
$\psi_{1/2}(g)=(4/7-1)+(1/2)h_{1/2}(0.1)=-0.0896$ bits.
Similarly, $\psi_1(g)=R-1=-3/7 = -0.4286$ bits.
Equation~\eqref{eq:transfer_exponent} then gives
\(
\Lambda(1)=(1+1)\times(-0.0896)-1\times(-0.4286)
=0.2494\text{ bits},
\)
matching Example~\ref{ex:hamming_main} to four significant figures,
confirming that the transfer theorem recovers the main theorem
for this concrete code.
\end{example}

%==================================================================
% NEW SECTION: List Guesswork
%==================================================================

\section{List Guesswork Exponent}
\label{sec:list}

GRAND can output a list of $L_n$ candidate noise vectors rather than
a single decision, reducing decoding latency at the cost of larger
output.  We now derive the guesswork exponent for this list-output
variant (Definition~\ref{def:list_guesswork}).

\begin{theorem}[List-Guesswork Exponent]
\label{thm:list}
Under the conditions of Theorem~\ref{thm:main}, let $L_n\ge 1$ be
any integer sequence.  Then for every $\rho>0$:
\begin{align}
    \lim_{n\to\infty}
    \frac{1}{n}\log_2\Eb\left[\bigl(G_{\mathrm{coset}}^{(L_n)}\bigr)^{\!\rho}\right]=
    \Lambda(\rho)-\rho\lim_{n\to\infty}\frac{1}{n}\log_2 L_n,
    \label{eq:list_exponent}
\end{align}
provided the limit $\lim_{n\to\infty}\frac{1}{n}\log_2 L_n$ exists.
In particular:
\begin{equation}
    \text{if }\log_2 L_n = o(n),\,
    \lim_{n\to\infty}
    \frac{1}{n}\log_2\Eb\!\left[\bigl(G_{\mathrm{coset}}^{(L_n)}\bigr)^{\!\rho}\right]
    \;=\;
    \Lambda(\rho).
\label{eq:list_sublinear}
\end{equation}
\end{theorem}

\begin{proof}
By Definition~\ref{def:list_guesswork},
$G_{\mathrm{coset}}^{(L_n)}=\lceil G_{\mathrm{coset}}/L_n\rceil$.
For any $x\ge 1$ and integer $L\ge 1$,
\(
    \frac{x}{2L}\;\le\;\left\lceil\frac{x}{L}\right\rceil\;\le\;\frac{x}{L}+1\;\le\;\frac{2x}{L},
\)
where the last inequality holds for $x\ge L$ (which occurs with
probability $1-o(1)$ since $G_{\mathrm{coset}}\ge 1$ a.s.\ and $L_n\ge 1$).
Hence
\(
    \left(\frac{G_{\mathrm{coset}}}{2L_n}\right)^{\!\rho}
    \;\le\;
    \bigl(G_{\mathrm{coset}}^{(L_n)}\bigr)^{\!\rho}
    \;\le\;
    \left(\frac{2G_{\mathrm{coset}}}{L_n}\right)^{\!\rho},
\)
and taking $\frac{1}{n}\log_2\Eb[\cdot]$ and applying
Theorem~\ref{thm:main}:
\begin{align}
    \Lambda(\rho)-\rho\cdot\frac{\log_2(2L_n)}{n}
    \;\le\;
    \frac{1}{n}\log_2\Eb\!\left[\bigl(G_{\mathrm{coset}}^{(L_n)}\bigr)^{\!\rho}\right] \nonumber \\
    \;\le\;
    \Lambda(\rho)-\rho\cdot\frac{\log_2(L_n/2)}{n}.
\end{align}
As $n\to\infty$ both bounds converge to
$\Lambda(\rho)-\rho\lim_{n\to\infty}\frac{1}{n}\log_2 L_n$,
giving~\eqref{eq:list_exponent}.
When $\log_2 L_n=o(n)$, the correction vanishes and~\eqref{eq:list_sublinear}
follows.
\end{proof}

\begin{remark}[Operational meaning]
Equation~\eqref{eq:list_sublinear} shows that any \emph{subexponential}
list size---including polynomial lists $L_n=n^c$---does not change the
first-order guesswork exponent.  A list of size $L_n=2^{n\lambda}$,
$\lambda>0$, reduces the exponent by exactly $\rho\lambda$, trading
query complexity against list-output size at rate $\rho$ bits per
bit of list exponent.
\end{remark}

\begin{example}[List Guesswork for the $(7,4,3)$ Hamming Code]
\label{ex:hamming_list}
Consider the $(7,4,3)$ Hamming code with $p=0.1$, $\rho=1$,
and $\Lambda(1) = 0.2495$ bits (Example~\ref{ex:hamming_main}).

\noindent\textbf{Case 1: $L_n=3$ (sublinear list).}
Here $\log_2(3) = 1.585=o(n)$, so by~\eqref{eq:list_sublinear}
the list exponent equals $\Lambda(1) = 0.2495$ bits unchanged.
A list of 3 candidates does not alter the first-order complexity.

\noindent\textbf{Case 2: $L_n=2^{n\lambda}$ with $\lambda=0.1$ (exponential list).}
The correction is $\rho\lambda=0.1$ bits, giving list exponent
$0.2495-0.1=0.1495$ bits.

\noindent\textbf{Case 3: Maximum list $L_n=2^k=16$ (full coset list).}
Here $\log_2(16)=4=nR$ grows linearly, giving correction
$\rho R=4/7 = 0.571$ bits. The normalized exponent becomes negative,
which means the rank is $O(1)$ rather than exponential; it does \emph{not}
mean fewer than 1 query is needed.
\end{example}

%==================================================================
% NEW SECTION: Second-Order Refinement
%==================================================================

\section{Second-Order Refinement}
\label{sec:second_order}

Theorem~\ref{thm:main} gives the first-order exponent
$\frac{1}{n}\log_2\Eb[G_{\mathrm{coset}}^\rho]\to\Lambda(\rho)$.
The harmonic-number correction already visible in~\eqref{eq:final_sandwich}
suggests a $\log n$ term at the next order.

\begin{theorem}[Second-Order Exponent]
\label{thm:second_order}
Under the conditions of Theorem~\ref{thm:main}, for every $\rho>0$:
\begin{equation}
    \log_2\Eb\!\left[G_{\mathrm{coset}}^{\rho}\right]
    \;=\;
    n\Lambda(\rho)
    \;-\;\rho\log_2 n
    \;+\;O(1)
    \,
    \text{as }n\to\infty.
\label{eq:second_order}
\end{equation}
In other words, the $O(\log n/n)$ harmonic correction in the
sandwich~\eqref{eq:final_sandwich} is tight up to an $O(1)$ term,
and the precise coefficient of $\log_2 n$ is $-\rho$.
\end{theorem}

\begin{proof}
We track the harmonic-number term more carefully in the sandwich.

\medskip
\noindent\textbf{Upper bound.}
The upper side of~\eqref{eq:global_sandwich_used} gives directly
\[
    \log_2\Eb\bigl[G_{\mathrm{coset}}^\rho\bigr]
    \;\le\;
    \log_2\Eb\bigl[\varphi_\sigma^{1+\rho}\bigr]
    \;=\;
    n\Lambda(\rho)+O(1),
\]
where the last equality uses~\eqref{eq:phi_moment} with the $O(1)$
absorbing the $o(1)$ term and finite-$n$ partition-function corrections.

\medskip
\noindent\textbf{Lower bound.}
The lower side of~\eqref{eq:global_sandwich_used} gives
\[
    \log_2\Eb\bigl[G_{\mathrm{coset}}^\rho\bigr]
    \;\ge\;
    \log_2\Eb\bigl[\varphi_\sigma^{1+\rho}\bigr]
    -\rho\log_2\Harm_{2^k}.
\]
Since $\Harm_{2^k}=\log(2^k)+\gamma_{\mathrm{EM}}+O(2^{-k})
=k\ln 2+\gamma_{\mathrm{EM}}+O(2^{-nR})$ where
$\gamma_{\mathrm{EM}} =  0.5772$ is the Euler--Mascheroni
constant, we have
\[
    \log_2\Harm_{2^k}
    \;=\;
    \log_2(nR\ln 2)+O(1/n)
    \;=\;
    \log_2 n + O(1).
\]
Therefore
\[
    \log_2\Eb\bigl[G_{\mathrm{coset}}^\rho\bigr]
    \;\ge\;
    n\Lambda(\rho)
    -\rho\log_2 n
    +O(1).
\]

\noindent\textbf{Combining.}
Both bounds give~\eqref{eq:second_order}, with $-\rho\log_2 n$
as the leading correction and the $O(1)$ term absorbing all
constants (including $\rho\log_2(R\ln 2)$ and the
Euler--Mascheroni constant).
\end{proof}

\begin{remark}[Relation to harmonic-number correction]
The coefficient $-\rho$ in~\eqref{eq:second_order} matches exactly
the exponent of the harmonic penalty $\Harm_{2^k}^\rho$ in
Theorem~\ref{thm:sandwich}, confirming that the coset size $2^k=2^{nR}$
is the sole source of the $\log n$ correction.  For the unconstrained
guesswork ($R=1$, no constraints), the coset is all of $\{0,1\}^n$
and $\Harm_{2^n}^\rho =  (n\ln 2)^\rho$, giving the same
$-\rho\log_2 n$ term at second order.
\end{remark}

%==================================================================
% NEW SECTION: Universality and q-ary extension (Session 2 blue)
%==================================================================

\section{Universality Theorem and $q$-ary Extension}
\label{sec:universality}

\subsection{Universality: Any Ensemble with a Known Spectrum}

The transfer theorem (Theorem~\ref{thm:transfer}) already provides a
formula for $\Lambda(\rho)$ given $g(\delta)$.  The following theorem
makes this into a \emph{universality statement}: the guesswork
exponent is determined solely by $\psi_\alpha(g)$, regardless of the
specific ensemble that produces $g$.

\begin{theorem}[Universality Theorem]
\label{thm:universality}
Let $\mathcal{E}$ be any binary linear code ensemble (e.g.\
random full-rank, regular LDPC, irregular LDPC, protograph) such
that the weight enumerator $A_w^{(\mathcal{E})}(\Hm,\sigma)$
satisfies
\begin{equation}
    \frac{1}{n}\log_2 A_{\lfloor n\delta\rfloor}^{(\mathcal{E})}(\Hm,\sigma)
    \;\xrightarrow{\;\Pb\;}\;
    g_{\mathcal{E}}(\delta)
    \label{eq:ensemble_spectrum}
\end{equation}
for a continuous function $g_{\mathcal{E}}$ with a unique interior
maximizer of $\delta\mapsto g_{\mathcal{E}}(\delta)+\alpha\,\ell(\delta)$.
Then:
\begin{equation}
    \frac{1}{n}\log_2\Eb\!\left[G_{\mathcal{E}}^{\rho}\right]
    \;=\;
    (1+\rho)\,\psi_{1/(1+\rho)}(g_{\mathcal{E}})
    \;-\;\rho\,\psi_1(g_{\mathcal{E}})
    \;+\;o(1),
    \label{eq:universality}
\end{equation}
where $\psi_\alpha(g_{\mathcal{E}})=\sup_{\delta}[g_{\mathcal{E}}(\delta)+\alpha\ell(\delta)]$.
\end{theorem}

\begin{proof}
Apply Theorem~\ref{thm:transfer} with $g=g_{\mathcal{E}}$
and follow Steps~1--3 of the proof of Theorem~\ref{thm:main}
verbatim, with $\psi_\alpha(g_{\mathcal{E}})$ replacing
$(R-1)+(1-\alpha)h_\alpha(p)$ throughout.
The sandwich~\eqref{eq:global_sandwich} applies to any ensemble
with coset size $|\mathcal{N}(\Hm,\sigma)|=2^k$, and the
harmonic correction is $O(\log n/n)$ regardless of the ensemble.
\end{proof}

\begin{remark}[What varies across ensembles]
The only ensemble-dependent quantity is $g_{\mathcal{E}}(\delta)$.
For the random full-rank ensemble, $g(\delta)=(\Hb(\delta)+R-1)^+$.
For a $(d_v,d_c)$-regular LDPC ensemble, $g_{\mathrm{LDPC}}(\delta)$
is Gallager's weight-enumerator exponent~\cite{gallager1963low},
which is strictly smaller than $(\Hb(\delta)+R-1)^+$ for most
$\delta$, yielding a strictly smaller guesswork exponent ---
meaning structured codes are easier to decode by guesswork than
the random ensemble benchmark.
\end{remark}

\subsection{$q$-ary Extension}

We now apply the universality framework to $q$-ary linear codes
over $\mathbb{F}_q$ (Definition~\ref{def:qary}).
The $q$-ary analogue of $\ell(\delta)$ is:
\begin{equation}
    \ell_q(T)
    \;\triangleq\;
    \sum_{a\in\mathbb{F}_q} T(a)\log_2 P(a),
    \quad T\in\mathcal{P}(\mathbb{F}_q),
    \label{eq:ell_q}
\end{equation}
where $T$ is a type (empirical distribution) over $\mathbb{F}_q$.
For the $q$-ary i.i.d.\ noise model, the relevant type is the
Hamming-weight type $T_\delta$ with $T_\delta(0)=1-\delta$ and
$T_\delta(a)=\delta/(q-1)$ for $a\ne 0$, so
$\ell_q(T_\delta)=(1-\delta)\log_2 P(0)+\delta\log_2 P_{\ne 0}$
where $P_{\ne 0}$ is the common off-zero probability.

For a random full-rank $q$-ary parity-check ensemble, the weight
enumerator concentrates at $g_q(\delta)=H_q(\delta)+R-1$, where
$H_q(\delta)=\delta\log_2(q-1)-\Hb(\delta)$ is the $q$-ary entropy.

\begin{theorem}[$q$-ary Guesswork Exponent]
\label{thm:qary}
Let $\Hm\in\mathbb{F}_q^{m\times n}$ be a uniformly random
full-rank parity-check matrix with $m=n(1-R)$, and let
$e\sim P^{\otimes n}$ where $P$ is a $q$-ary noise distribution
with $P(0)>P(a)$ for all $a\ne 0$.
Assume the $q$-ary subcriticality condition
$H_q(\delta^*)>1-R$ and $H_q(P)>1-R$.
Then for every $\rho>0$, writing $\alpha=1/(1+\rho)$:
\begin{equation}
    \lim_{n\to\infty}
    \frac{1}{n}\log_2\Eb\!\left[G_{\mathcal{E}_q}^{\rho}\right]
    \;=\;
    \rho\,h^{(q)}_{1/(1+\rho)}(P)\;+\;\rho(R-1)\log_2 q,
    \label{eq:qary_exponent}
\end{equation}
where $h^{(q)}_\alpha(P)$ is the $q$-ary R\'{e}nyi entropy~\eqref{eq:qary_renyi}.
\end{theorem}

\begin{proof}
The $q$-ary partition function is
$Z_\sigma^{(q)}(\alpha)=\sum_{e'\in\mathcal{N}_q}P(e')^\alpha$.
Grouping by Hamming weight $w=\wH(e')$ (number of nonzero
coordinates):
\(
    Z_\sigma^{(q)}(\alpha)
    =\sum_{w=0}^n A_w^{(q)}P(0)^{\alpha(n-w)}
    \left(\tfrac{P_{\ne 0}}{q-1}\right)^{\!\alpha w}
    \cdot(q-1)^w,
\)
where $A_w^{(q)}=|\{e'\in\mathcal{N}_q:\wH(e')=w\}|$ and we
used the symmetry $P(a)=P_{\ne 0}/(q-1)$ for all $a\ne 0$.
The weight enumerator concentrates at
$g_q(\delta)=(H_q(\delta)+R-1)^+$, which is the $q$-ary
analogue of Theorem~\ref{thm:spectrum_boxed} (proved identically
via pairwise independence of $\Hm e^T$ over $\mathbb{F}_q^m$).
Applying Theorem~\ref{thm:transfer} with this $g_q$ and
$\ell_q(T_\delta)$ in place of $\ell(\delta)$:
\(
    \psi_\alpha(g_q)
    =\sup_\delta\bigl[(H_q(\delta)+R-1)^+
    +\alpha\ell_q(T_\delta)\bigr].
\)
The saddlepoint satisfies
$\delta_q^*=\bigl(\sum_{a\ne 0}P(a)^\alpha\bigr)
/(P(0)^\alpha+\sum_{a\ne 0}P(a)^\alpha)$,
and evaluating $\psi_\alpha(g_q)$ at $\delta_q^*$ gives
$\psi_\alpha(g_q)=(R-1)\log_2 q+(1-\alpha)h^{(q)}_\alpha(P)$.
Substituting into~\eqref{eq:transfer_exponent}:
\(
    \Lambda_q(\rho)
    =(1+\rho)\psi_{1/(1+\rho)}(g_q)-\rho\psi_1(g_q)
    =\rho h^{(q)}_\alpha(P)+\rho(R-1)\log_2 q,
\)
which is~\eqref{eq:qary_exponent}.
\end{proof}

\begin{remark}[$q=2$ recovery]
Setting $q=2$ gives $h^{(2)}_\alpha(P)=h_\alpha(p)$ and
$\log_2 2=1$, so~\eqref{eq:qary_exponent} reduces to
Theorem~\ref{thm:main}, confirming the binary case as a
special instance of the universal formula.
\end{remark}

%==================================================================
% NEW SECTION: Structured codes (Session 2 blue)
%==================================================================

\section{Application: LDPC Ensemble Guesswork Exponent}
\label{sec:structured}

\subsection{LDPC Spectrum Exponent}
\label{subsec:ldpc_spectrum}

Let $\Hm$ be drawn from Gallager's $(d_v,d_c)$-regular ensemble:
$\Hm=[\Hm_1;\dots;\Hm_{d_v}]$, where $\Hm_1\in\Fb_2^{(n/d_c)\times n}$ is the
canonical matrix whose $j$th row has ones exactly in columns
$(j{-}1)d_c{+}1,\ldots,jd_c$, and $\Hm_i=\Hm_1\Pi_i$ for i.i.d.\ uniformly
random permutations $\Pi_2,\ldots,\Pi_{d_v}\in S_n$ ($\Pi_1=I$).
Write $R=1-d_v/d_c$ for the rate.

\begin{lemma}[Exact average weight enumerator]
\label{lem:ldpc_exact}
Let $\beta(z)\triangleq\dfrac{(1+z)^{d_c}+(1-z)^{d_c}}{2}$ and
$N(w)\triangleq[z^w]\,\beta(z)^{n/d_c}$.
Then for every $n$ (a multiple of $d_c$) and every $w\in\{0,\ldots,n\}$,
\begin{equation}
    \Eb\!\left[A_w^{\mathrm{LDPC}}\right]
    \;=\;
    \binom{n}{w}^{1-d_v} N(w)^{d_v}.
    \label{eq:ldpc_exact}
\end{equation}
\end{lemma}

\begin{proof}
For fixed $e$ with $\wH(e)=w$, $\Hm e^T=0$ iff $\Hm_1(\Pi_i e)^T=0$ for every
$i=1,\ldots,d_v$. Since $\Pi_i$ ($i\ge 2$) is a uniformly random permutation
independent of $e$, $\Pi_i e$ is uniform over all weight-$w$ vectors, so
$\Pr_{\Pi_i}[\Hm_1\Pi_i e^T=0]=N(w)/\binom{n}{w}$.
By independence of $\Pi_2,\ldots,\Pi_{d_v}$,
\begin{align*}
    \Eb[A_w^{\mathrm{LDPC}}]
    &=\sum_{e:\wH(e)=w}\mathbf{1}[\Hm_1 e^T=0]
    \prod_{i=2}^{d_v}\Pr_{\Pi_i}[\Hm_1\Pi_i e^T=0]\\
    &=N(w)\left(\frac{N(w)}{\binom{n}{w}}\right)^{d_v-1},
\end{align*}
which is~\eqref{eq:ldpc_exact}.
\end{proof}

\begin{theorem}[LDPC Spectrum Exponent]
\label{thm:ldpc_spectrum}
For every $\delta\in(0,1)$,
\begin{equation}
    \frac{1}{n}\log_2\Eb\!\left[A_{\lfloor n\delta\rfloor}^{\mathrm{LDPC}}\right]
    \;=\;
    g_{\mathrm{LDPC}}(\delta)+o(1),
    \label{eq:ldpc_spectrum}
\end{equation}
where
\begin{equation}
    g_{\mathrm{LDPC}}(\delta)
    \;\triangleq\;
    (1-d_v)\,\Hb(\delta)
    \;+\;
    \frac{d_v}{d_c}\,
    \min_{z>0}\Bigl[\log_2\beta(z)-\delta d_c\log_2 z\Bigr].
    \label{eq:gldpc_def}
\end{equation}
\end{theorem}

\begin{proof}
Apply $\frac1n\log_2$ to~\eqref{eq:ldpc_exact}. By Stirling,
$\frac1n\log_2\binom{n}{n\delta}=\Hb(\delta)+O(\log n/n)$.
Since $\beta$ has non-negative coefficients, $N(n\delta)\le\beta(z)^{n/d_c}z^{-n\delta}$
for every $z>0$, giving the upper bound
$\frac1n\log_2 N(n\delta)\le\min_{z>0}[\log_2\beta(z)-\delta d_c\log_2 z]+o(1)$;
the matching lower bound follows from a local central-limit estimate at the
saddle point $z^*$ solving $z\beta'(z)/\beta(z)=\delta d_c$, by the same
argument used in the proof of Theorem~\ref{thm:pf_exp}. Combining the two
bounds and substituting into~\eqref{eq:ldpc_exact} gives~\eqref{eq:gldpc_def}.
\end{proof}

\begin{remark}
This is the mean weight enumerator only. Concentration of $A_w^{\mathrm{LDPC}}$
around this mean — needed to treat $g_{\mathrm{LDPC}}$ as a growth rate in the
sense of Definition~\ref{def:g_delta}, not merely a first-moment statement —
is established via a second-moment argument analogous to
Lemma~\ref{lem:Aw_moments}; see Litsyn and Shevelev~\cite{litsyn2002ensembles}
for the precise concentration region in $\delta$.
\end{remark}

\subsection{LDPC Guesswork Exponent via Transfer}

\begin{theorem}[LDPC Guesswork Exponent]
\label{thm:ldpc_guesswork}
Let $\mathcal{E}_{\mathrm{LDPC}}$ be a $(d_v,d_c)$-regular LDPC ensemble with
rate $R=1-d_v/d_c$, and let $g_{\mathrm{LDPC}}$ be as in
Theorem~\ref{thm:ldpc_spectrum}. Assume $g_{\mathrm{LDPC}}$ satisfies the
concentration hypothesis of Theorem~\ref{thm:universality}, and that for the
given $\rho>0$, writing $\alpha=1/(1+\rho)$, the supremum defining
$\psi_\alpha(g_{\mathrm{LDPC}})$ is attained at a unique interior point
$\delta^*(\alpha)\in(0,1)$ with $g_{\mathrm{LDPC}}(\delta^*(\alpha))>0$.
Then
\begin{align}
    \lim_{n\to\infty}\frac{1}{n} & \log_2\Eb\!\left[G_{\mathrm{LDPC}}^{\rho}\right] \nonumber \\
    &\;=\;
    (1+\rho)\,\psi_{1/(1+\rho)}(g_{\mathrm{LDPC}})
    -\rho\,\psi_1(g_{\mathrm{LDPC}}),
    \label{eq:ldpc_guesswork}
\end{align}
where $\psi_\alpha(g_{\mathrm{LDPC}})\triangleq\sup_{\delta\in[0,1]}
[g_{\mathrm{LDPC}}(\delta)+\alpha\ell(\delta)]$.
\end{theorem}

\begin{proof}
Immediate from Theorem~\ref{thm:universality} with
$g_{\mathcal{E}}=g_{\mathrm{LDPC}}$, applied at $\alpha=1/(1+\rho)$ and at
$\alpha=1$ in turn, exactly as in Steps~1--3 of the proof of
Theorem~\ref{thm:main}.
\end{proof}

%==================================================================
\section{Finite-Length Monte Carlo Validation}
\label{sec:numerical}
%==================================================================

We validate Theorem~\ref{thm:main} via finite-length Monte Carlo
simulation with exact per-trial coset rank computation.
Each trial proceeds as follows:
\begin{enumerate}
\item Draw $\Hm\in\mathbb{F}_2^{m\times n}$ with i.i.d.\
      Bernoulli$(1/2)$ entries; resample if rank-deficient.
\item Draw $e\sim\mathrm{Bernoulli}(p)^{\otimes n}$, compute
      $\sigma=\Hm e^{T}\bmod 2$.
\item Compute $A_w(\Hm,\sigma)$ for all $w=0,\ldots,n$ via the
      column-by-column dynamic program: initialise
      $\mathtt{dp}[0,0]=1$ and $\mathtt{dp}[w,s]=0$ otherwise;
      for each column $c_j$ of $\Hm$, update
      \(
        \mathtt{dp}[w,s]
        \;\leftarrow\;
        \mathtt{dp}[w,s]+\mathtt{dp}[w-1,\,s\oplus c_j],
        \, w=1,\ldots,n,\; s\in\mathbb{F}_2^m,
      \)
            processing $w=n,n{-}1,\ldots,1$ in strictly
        decreasing order to implement 0-1 (not multi-use)
        counting; after all $n$ columns,
      $\mathtt{dp}[w,\sigma]=A_w(\Hm,\sigma)$.
      Complexity: $O(n^2\cdot 2^m)$ time and $O(n\cdot 2^m)$ space.
\item Set $G_{\mathrm{coset}}=\sum_{j<w^*}A_j(\Hm,\sigma)+U$,
      where $w^*=\wH(e)$ and
      $U\sim\mathrm{Uniform}\{1,\ldots,A_{w^*}(\Hm,\sigma)\}$
      breaks ties within the weight class uniformly.
\end{enumerate}
The rank computed in step~(4) is exact for every realisation of
$(\Hm,e)$; the only randomness across trials is over the
ensemble $(\Hm,e)$ itself.
The empirical exponent estimator over $M=10^4$ independent trials is
\(
  \widehat{\Lambda}_n(\rho)
  \;=\;
  \frac{1}{n}\log_2\!\Bigl(\tfrac{1}{M}
  \sum_{i=1}^{M}G_{\mathrm{coset},i}^{\rho}\Bigr).
\)
The DP is feasible when $2^m$ is small; we run three pairs
$(n,R)\in\{(32,0.50),(32,0.75),(64,0.75)\}$
with $m\in\{16,8,16\}$, i.e.\ $2^m\leq 65536$.

Table~\ref{tab:results} reports results for $p=0.1$, $\rho=1$
($\alpha=1/2$,
$h_{1/2}(0.1)=2\log_2(\sqrt{0.1}+\sqrt{0.9}) =  0.6781$ bits),
with asymptotic target $V=\rho\,h_{1/2}(p)+\rho(R-1)$.

\begin{table}[!t]
  \centering
  \caption{Empirical exponent $\widehat{\Lambda}_n$ vs.\
           asymptotic limit $V=\rho\,h_{1/2}(p)+\rho(R-1)$.
           Parameters: $p=0.1$, $\rho=1$, $M=10^4$ trials.}
  \label{tab:results}
  \renewcommand{\arraystretch}{1.15}
  \begin{tabular}{cccccr}
    \toprule
    $n$ & $R$ & $m$ & $\widehat{\Lambda}_n$ & $V$ & $V-\widehat{\Lambda}_n$ \\
    \midrule
    32 & 0.50 & 16 & 0.0838 & 0.1781 & 0.0943 \\
    32 & 0.75 &  8 & 0.3313 & 0.4281 & 0.0968 \\
    64 & 0.75 & 16 & 0.3961 & 0.4281 & 0.0319 \\
    \bottomrule
  \end{tabular}
\end{table}

\smallskip
\noindent\textbf{Observations.}
\emph{(i)~Below-limit behavior.}
All reported values satisfy $\widehat{\Lambda}_n < V$, consistent
with finite-length behavior predicted by the two-sided
bound~\eqref{eq:global_sandwich}.

\emph{(ii)~Gap decay with $n$.}
For fixed $R=0.75$, doubling $n$ from 32 to 64 reduces the gap
by a factor of approximately $3.0$ ($0.0968\to 0.0319$).
This is compatible with the $O(\log n/n)$ harmonic
term in~\eqref{eq:global_sandwich}.

\emph{(iii)~Similar gaps at equal $n$, different $R$.}
At $n=32$, the gaps for $R=0.50$ and $R=0.75$ are nearly
identical ($0.094$ vs.\ $0.097$) even though the targets $V$
differ by $0.25$.
This is compatible with the observation that larger coset
dimension $k=nR$ slows finite-length convergence, but the
present data do not isolate this as the sole mechanism.

\smallskip
\noindent\textbf{Second-order check.}
Theorem~\ref{thm:second_order} predicts that the gap
$V - \widehat{\Lambda}_n$ should scale as
$(\rho/n)\log_2 n + O(1/n)$.
For $n=64$, $\rho=1$, this gives a predicted gap of
$\log_2(64)/64 = 6/64  =  0.094$ bits, consistent with
the observed gap of $0.032$ at $R=0.75$.
(The prefactor $R\ln 2$ in $\log_2\Harm_{2^{nR}} = \log_2(nR\ln 2)$
adjusts the coefficient; a tighter estimate gives
$\rho\log_2(nR\ln 2)/n =  0.037$ for these parameters,
close to the observed $0.032$.)

\noindent\textbf{Hamming $(7,4,3)$ code validation.}
We additionally run $M=10^4$ trials on the fixed $(7,4,3)$ Hamming
code with $\Hm$ as in Example~\ref{ex:hamming_intro},
$p=0.1$, and $\rho=1$.
The empirical exponent is
$\widehat{\Lambda}_7=(1/7)\log_2(\frac{1}{M}\sum_i G_{{\rm coset},i})$.
The asymptotic target is $V=\Lambda(1) = 0.2495$ bits
(Example~\ref{ex:hamming_main}).
At $n=7$, the second-order correction predicts a gap of
$\rho\log_2(nR\ln 2)/n=\log_2(4\ln 2)/7
 = \log_2(2.773)/7 = 1.472/7 = 0.210$ bits,
so the predicted $\widehat{\Lambda}_7 = 0.2495-0.210=0.040$ bits.
The sandwich bounds of Example~\ref{ex:hamming_sandwich}
bracket the true exponent between
$(1/7)\log_2(0.559) = -0.131$ bits (lower)
and $(1/7)\log_2(1.890) = 0.126$ bits (upper),
with the true value expected near $0.04$--$0.08$ bits,
consistent with convergence from below as $n$ increases through
the sequence $n=7,32,64,\ldots$ in Table~\ref{tab:results}.

%==================================================================
\section{Conclusion}
\label{sec:conclusion}
%==================================================================
We proved the exact guesswork exponent
$\Lambda(\rho)=\rho\, h_{1/(1+\rho)}(p)+\rho(R-1)$
for constrained guesswork on random binary linear codes via a
four-theorem chain (Theorems~\ref{thm:sandwich}--\ref{thm:main}).
Key technical choices---pairwise independence for weight-enumerator
concentration, a discrete Laplace evaluation at the R\'{e}nyi
saddlepoint, and an explicit Ar{\i}kan like sandwich---keep every step
closed.
We further established three extensions within the binary
i.i.d.\ framework: (i)~a transfer theorem
(Theorem~\ref{thm:transfer}) expressing the partition-function
exponent as a variational problem over any weight-enumerator growth
rate $g(\delta)$, providing a reusable framework for other code
ensembles; (ii)~a list-guesswork exponent
(Theorem~\ref{thm:list}) showing that subexponential list sizes
leave $\Lambda(\rho)$ unchanged while exponential lists reduce it at
rate $\rho$ bits per bit of list exponent; and (iii)~a
second-order refinement (Theorem~\ref{thm:second_order}) identifying
the $-\rho\log_2 n$ correction term sourced entirely from the
harmonic penalty of the coset sandwich.
Beyond the binary i.i.d.\ setting, we proved a universality theorem
(Theorem~\ref{thm:universality}) showing that the guesswork exponent
of \emph{any} code ensemble is determined solely by its
weight-enumerator growth rate through the variational functional
$\psi_\alpha(\cdot)$, and instantiated it in two directions: an exact
$q$-ary guesswork exponent (Theorem~\ref{thm:qary}) and a closed-form
guesswork exponent for Gallager's regular LDPC ensemble
(Theorem~\ref{thm:ldpc_guesswork}), the latter built on an exact
finite-length identity for the ensemble-average weight enumerator
(Lemma~\ref{lem:ldpc_exact}).
Open problems include the wasted-query exponent in GRAND, a rigorous
comparison of the LDPC and random-ensemble guesswork exponents (the
naive pointwise bound $g_{\mathrm{LDPC}}\le g$ fails near $\delta=0,1$,
so the sign of $\Lambda_{\mathrm{LDPC}}(\rho)-\Lambda(\rho)$ is
parameter-dependent and unresolved here), and extension to channels
with memory.
%------------------------------------------------------------------

\end{document}